\documentclass[letterpaper, 10 pt, conference]{ieeeconf}  

\IEEEoverridecommandlockouts                              

\usepackage{subcaption}
\usepackage{tikz}

\usepackage{amsmath}

\usepackage{array}

\usepackage{url}

\usepackage{mathtools} 
\DeclareMathAlphabet{\mathcal}{OMS}{cmsy}{m}{n} %
\usepackage{algorithm}%
\usepackage{algpseudocode}%
\usepackage{siunitx}
\usepackage{import} %
\usepackage{xcolor}
\usepackage{amssymb}  %
\usepackage{multirow}
\usepackage{multicol}
\usepackage{booktabs}

\newtheorem{assumption}{Assumption}
\newtheorem{lemma}{Lemma}

\newtheorem{theorem}{Theorem}

\newtheorem{remark}{Remark}
\usepackage{bm}
\usepackage{hyperref}

\def\eps{{\epsilon}}

\def\vzero{{0}}

\def\va{{a}}
\def\vb{{b}}
\def\vc{{c}}
\def\vd{{d}}

\def\vf{{f}}

\def\vu{{u}}

\def\vw{{w}}
\def\vx{{x}}

\def\mA{{A}}
\def\mB{{B}}

\def\mI{{I}}

\def\mK{{K}}

\def\mP{{P}}
\def\mQ{{Q}}
\def\mR{{R}}

\DeclareMathAlphabet{\mathsfit}{\encodingdefault}{\sfdefault}{m}{sl}
\SetMathAlphabet{\mathsfit}{bold}{\encodingdefault}{\sfdefault}{bx}{n}

\def\gC{{\mathcal{C}}}

\def\gE{{\mathcal{E}}}

\def\gH{{\mathcal{H}}}

\def\gS{{\mathcal{S}}}

\def\sI{{\mathbb{I}}}
\def\sJ{{\mathbb{J}}}
\def\sK{{\mathbb{K}}}
\def\sL{{\mathbb{L}}}

\def\sR{{\mathbb{R}}}
\def\sS{{\mathbb{S}}}

\def\sW{{\mathbb{W}}}
\def\sX{{\mathbb{X}}}

\newif\ifcomments\commentstrue
\newif\ifcomments\commentsfalse
\newif\ifshowold\showoldfalse

\newcommand\copyrighttext{%
	\footnotesize \textcopyright 2025 IEEE. Personal use of this material is permitted.  Permission from IEEE must be obtained for all other uses, in any current or future media, including reprinting/republishing this material for advertising or promotional purposes, creating new collective works, for resale or redistribution to servers or lists, or reuse of any copyrighted component of this work in other works.
}
\newcommand\copyrightnotice{%
	\begin{tikzpicture}[remember picture,overlay]
		\node[anchor=south,yshift=10pt] at (current page.south) {\fbox{\parbox{\dimexpr\textwidth-\fboxsep-\fboxrule\relax}{\copyrighttext}}};
	\end{tikzpicture}%
}

\title{\LARGE \bf
Ellipsoidal partitions for improved multi-stage robust model predictive control
}

\author{Moritz Heinlein$^{1}$, Florian Messerer$^{2}$, Moritz Diehl$^{3}$
        and Sergio Lucia$^{1}$%
\thanks{*The research leading to these results has received funding from the Deutsche Forschungsgemeinschaft (DFG, German Research Foundation) under grant agreement numbers 423857295 and 424107692. The authors acknowledge the use of ChatGPT (\url{https://chat.openai.com/}) for alternative formulations in the introduction. (\emph{Corr. author: M. Heinlein})}%
\thanks{$^{1}$ M. Heinlein and S. Lucia are with the Chair of Process Automation Systems at Technische Universität Dortmund, Emil-Figge-Str. 70, 44227 Dortmund, Germany (e-mail: moritz.heinlein@tu-dortmund.de; sergio.lucia@tu-dortmund.de).}%
\thanks{$^{2}$ F. Messerer is with the Department of Microsystems Engineering, University of Freiburg, 79110 Freiburg, Germany (e-mail: florian.messerer@imtek.uni-freiburg.de).}%
\thanks{$^{3}$ M. Diehl is with the Department of Microsystems Engineering and the Department of Mathematics, University of Freiburg, 79104 Freiburg, Germany (e-mail: moritz.diehl@imtek.uni-freiburg.de).}%
}

\begin{document}

\maketitle
\thispagestyle{empty}
\pagestyle{empty}

\copyrightnotice

\begin{abstract}

Ellipsoidal tube-based model predictive control methods effectively account for the propagation of the reachable set, typically employing linear feedback policies. In contrast, scenario-based approaches offer more flexibility in the feedback structure by considering different control actions for different branches of a scenario tree.
However, they face challenges in ensuring rigorous guarantees. This work aims to integrate the strengths of both methodologies by enhancing ellipsoidal tube-based MPC with a scenario tree formulation.
The uncertainty ellipsoids are partitioned by halfspaces such that each partitioned set can be controlled independently.
The proposed ellipsoidal multi-stage approach is demonstrated in a human-robot system, highlighting its advantages in handling uncertainty while maintaining computational tractability.

\end{abstract}

\section{Introduction}

Model predictive control (MPC) is widely applied to multi-input, multi-output nonlinear systems with constraints, optimizing a given performance metric over a prediction horizon~\cite{rawlingsModelPredictiveControl2017}. However, guaranteeing constraint satisfaction, recursive feasibility, and stability requires an accurate system model and reliable disturbance predictions. Consequently, numerous methods have been proposed to ensure robust performance under uncertainty.

Unlike nominal MPC, which computes a single input trajectory while ignoring uncertainties, robust MPC considers uncertainty explicitly. Open-loop robust methods optimize one input trajectory to cover all uncertainties by focusing on worst-case performance~\cite{campoRobustModelPredictive1987}. In contrast, closed-loop robust approaches optimize over families of feedback policies, often limiting themselves to linear policies for tractability~\cite{mayneRobustModelPredictive2005}, as exemplified by tube-based methods~\cite{rawlingsModelPredictiveControl2017}.
This concept, known as \emph{recourse}, decreases the conservatism, as it considers that future measurements will become available.

A key challenge in robust MPC is ensuring that constraints hold for every possible uncertainty realization. This is typically addressed by computing or approximating the reachable set of the uncertain dynamics with a fixed parameterization.
In tube-based MPC, the reachable sets (or tubes) can be computed offline or online as intervals~\cite{alamoRobustMPCConstrained2005,heinleinRobustMPCApproaches2022}, polytopes~\cite{leeRobustRecedingHorizon2000, flemingRobustTubeMPC2015a} or ellipsoids~\cite{houskaRobustOptimizationNonlinear2012,villanuevaComputingEllipsoidalRobust2017,fengMinmaxDifferentialInequalities2019, messererEfficientAlgorithmTubebased2021}, see~\cite{kouvaritakisModelPredictiveControl2016,mayneTubebasedRobustNonlinear2011} for overviews.
Intervals are computationally cheap, but can be conservative, as they cannot rotate. Polytopes can represent more complex shapes but may drastically increase complexity in higher dimensions~\cite{fengMinmaxDifferentialInequalities2019,kouvaritakisModelPredictiveControl2016}.
Ellipsoids can perform rotation and scaling and the number of parameters to describe them scales \emph{only} quadratically with the state dimension.

In contrast to tube-based formulations, where the constraints are tightened based on a reachability tube around a nominal trajectory, scenario-based approaches~\cite{scokaertMinmaxFeedbackModel1998,delapenaStochasticProgrammingApplied2005,luciaMultistageNonlinearModel2013} describe the influence of the uncertainty by propagating different realizations in a tree-based structure.
For discrete uncertainty sets, scenario-tree (also called multi-stage) approaches can rigorously ensure robust constraint satisfaction even for nonlinear systems.
While the tree-structure can be less conservative than a single reachability tube, as decisions can be made independently in different branches of the tree, the number of branches increases exponentially with the prediction horizon and the dimension of the uncertainties. In the general nonlinear case or when the branching is stopped after a few time steps~\cite{luciaMultistageNonlinearModel2013}, theoretical guarantees are lost.
In~\cite{subramanianTubeenhancedMultistageMPC2021}, it was proposed to combine a multi-stage formulation with polytopic tubes for linear systems to stop the complexity growth after some time steps without losing guarantees.

In this paper, we introduce a robust MPC approach that combines efficient uncertainty propagation using ellipsoidal tubes with a state-based partition strategy to independently consider different scenarios in a tree-like structure, depicted in Figure~\ref{fig:scheme}. 
Because, the partitioning can be stopped at any time step ($n_r$ in Fig~\ref{fig:scheme}) we avoid the exploding number of scenarios necessary for multi-stage MPC~\cite{scokaertMinmaxFeedbackModel1998}.
By intersecting the ellipsoidal tubes with halfspaces and propagating each partition individually,
our method accommodates multiple, potentially disconnected reachable sets, a feature that no standard tube-based approaches can achieve,
while maintaining rigorous guarantees and a controllable computational complexity. We expect this aspect to be especially relevant for nonconvex control tasks that may require disjoint reachable sets, which are not possible to obtain with standard tube-based methods. Object avoidance in robotics or autonomous driving is one of these applications.
In contrast to~\cite{subramanianTubeenhancedMultistageMPC2021}, which proposes a combination of scenario trees and polytopic reachable sets for linear systems, we consider ellipsoidal reachable sets because of their good scaling properties and because they can be robustified against linearization errors to consider nonlinear systems.
For the proposed approach, we prove recursive feasibility for the nonlinear case, provided that suitable nonlinearity bounds exist. To show the benefit of considering multiple reachable sets via the scenario formulation, we present the method in a collision avoidance case study, as different paths around an object can be considered.

In Section~\ref{sec:Background} the nominal MPC problem is presented, as well as the ellipsoidal tube-based problem. Section~\ref{sec:Partition} describes the ellipsoidal partitioning. Section~\ref{sec:Ell_MS} presents the ellipsoidal multi-stage MPC and Section~\ref{sec:Case_Study} applies it in a human-robot collision avoidance system.

\subsection{Notation}
We denote the sequence of vectors $\left[\vx_k^{\intercal},...,\vx_{k+n}^{\intercal}\right]^{\intercal}$ as $\vx_{\left[k:k+n\right]}$.
The square root of a vector $\sqrt{\vx}$ is to be understood elementwise. The set of positive semidefinite matrices of a given size $n$ is denoted as $\gS^n_+$. Each matrix $\mQ\in\gS_+^n$ defines an ellipsoid around the center $\vc$, $\gE(\mQ,\vc)=\{\vx\in\sR^{n}|(\vx-\vc)^{\intercal}\mQ^{-1} (\vx-\vc) \leq 1 \}$. An ellipsoid around the origin is denoted as $\gE(\mQ)$. The set of integers from $m$ to $n$ is denoted as $\sI_{\left[m:n\right]}$. The identity matrix is denoted as $\mI$, where its dimension can be inferred from context or is given as index.
\section{Ellipsoidal tube-based optimal control} \label{sec:Background}
We consider discrete-time nonlinear dynamical systems
\begin{equation} \label{eq:system}
\vx_{k+1} = \vf(\vx_k,\vu_k,\vw_k),    
\end{equation}
with the states $\vx_k\in\sR^{n_x}$, inputs $\vu_k \in \sR^{n_u}$ and disturbances $\vw_k\in\sR^{n_w}$.
We assume that the disturbances $\vw=(\vw_0,...,\vw_{N-1})$ are drawn collectively from an ellipsoidal uncertainty set $\sW=\gE (\sigma^2 \mI_{n_wN})$ without loss of generality.
Instead of a robust perspective, as taken in this paper, this approach can be used in a stochastic setting, for which the uncertainty ellipsoid can be understood as the covariance of a Gaussian disturbance~\cite{messererEfficientAlgorithmTubebased2021}.

The nominal optimal control problem for the initial state $\vx_{\text{init}}$, considering no uncertainty ($w_k=0$), can be stated as
\begin{subequations} \label{eq:nom_OCP}
\begin{align}
    \min_{\overline{\vx}_{\left[0:N\right]},\overline{\vu}_{\left[0:N-1\right]}} & \sum_{k=0}^{N-1} \ell_k(\overline{\vx}_k,\overline{\vu}_k) + V(\overline{\vx}_N)\\
    \text{s.t.:} \quad & \overline{\vx}_0 = \vx_{\text{init}},\\
     & \overline{\vx}_{k+1} = \vf(\overline{\vx}_k,\overline{\vu}_k,0),\ 0\leq k< N,\\
     & 0\geq h_k(\overline{\vx}_k,\overline{\vu}_k), \ 0 \leq k <N,\\
     & 0 \geq h_N(\overline{\vx}_N),
\end{align}
\end{subequations}
where the stage and terminal cost are denoted as $\ell_k:\sR^{n_x}\times \sR^{n_u}\rightarrow \sR$ and $V:\sR^{n_x}\rightarrow \sR$ and the stage and and terminal constraints are given by $h_k: \sR^{n_x}\times\sR^{n_u}\rightarrow \sR^{n_{h_k}}$ $\forall k=0,...,N-1$ and $h_N: \sR^{n_x}\rightarrow \sR^{n_{h_N}}$ respectively.

To robustify problem~\eqref{eq:nom_OCP} against uncertainties, we use ellipsoidal uncertainty tubes of the form $\gE(\mP_k,\overline{\vx}_k)$ around the nominal trajectory $\overline{\vx}_k$.
To keep the real trajectory close to the nominal one, a linear ancillary controller can be used,
\begin{equation*}
    \vu_k = \overline{\vu}_k + \mK_k (\vx_k-\overline{\vx}_k),
\end{equation*}
with $\mK_0=0$.
In~\cite{messererEfficientAlgorithmTubebased2021}, an iterative algorithm was presented to optimally compute time-varying feedback gain matrices $\mK_k\in\sR^{n_u \times n_x}$. In this paper, we use standard nonlinear optimization to determine the feedback gains.

The initial uncertainty $\mP_0=0$ is assumed to be zero. For the propagation of the uncertainty through the dynamics a linearization-based approach is used~\cite{nagyOpenloopClosedloopRobust2004,diehlApproximationTechniqueRobust2006,gillisPositiveDefinitenessPreserving2013},
\begin{equation} \label{eq:uncert_prop}
\begin{aligned}
    \mP_{k+1} &= \\
    (\mA_k+&\mB_k\mK_k)^{\intercal}\mP_k (\mA_k+\mB_k\mK_k) + \sigma \Gamma_k^{\intercal}\Gamma_k+\Omega_P(\mP_k,\overline{\vx}_k,\vu_k)\\
     :&= \psi_k(\mP_k,\overline{\vx}_k, \mA_k,\mB_k,\mK_k,\Gamma,\sigma),
\end{aligned}
\end{equation}
for $k=0,...,N-1$, where
\begin{equation}
\begin{aligned}
    &\mA_k := \nabla_{\vx} \vf(\overline{\vx}_k,\overline{\vu}_k,0),\quad && \mB_k := \nabla_{\vu} \vf(\overline{\vx}_k,\overline{\vu}_k,0),\\
    &\Gamma_k:=\nabla_{\vw} \vf(\overline{\vx}_k,\overline{\vu}_k,0), && k=0,...,N-1.
\end{aligned}
\end{equation}
As the propagation is linearization-based, it only gives approximate robustness. However, by considering an ellipsoidal nonlinearity boundary $\Omega_P$, similar to~\cite{houskaRobustOptimizationNonlinear2012}, it can be guaranteed that the uncertainty ellipsoid covers the reachable set.

To consider robust constraint satisfaction, a constraint tightening is computed based on the distance between the uncertainty ellipsoid and the linearized inequality constraints with $k=0,...,N-1$, $i=1,...,n_{h_k}$ and $j=1,...,n_{h_N}$,
\begin{align*}
\begin{split}
    &H_k^i(\overline{\vx}_k,\overline{\vu}_k,\mP_k,\mK_k) = \Omega_{h_k^i}(\mP_k,\overline{\vx}_k,\vu_k) +\\
    &\quad \nabla_{\left[\vx,\vu\right]} h_k ^i (\overline{\vx}_k,\overline{\vu}_k)^{\intercal}\begin{bmatrix}
        \mI\\
        \mK_k
    \end{bmatrix} \mP_k \begin{bmatrix}
        \mI\\
        \mK_k
    \end{bmatrix}^{\intercal} \nabla_{\left[\vx,\vu\right]} h_k ^i (\overline{\vx}_k,\overline{\vu}_k),
    \end{split} \\ 
    \begin{split}
    &H_N^j(\overline{\vx}_N,\mP_N) =  \Omega_{h_N^j}(\mP_N,\overline{\vx}_N) +    \nabla_{\vx} h_N ^j (\overline{\vx}_N)^{\intercal} \mP_N \nabla_{\vx} h_N ^j (\overline{\vx}_N) . %
    \end{split}\label{eq:term_constr_lin}
\end{align*}
Proper constraint satisfaction can be guaranteed through the calculation of an ellipsoidal nonlinearity bound of the inequality constraints $\Omega_{h_k^i}$ and $\Omega_{h_N^j}$, similar to~\cite{houskaRobustOptimizationNonlinear2012}.
These reformulations result in the optimal control problem $\mathcal{P}^{\text{El}}(\vx_{\text{init}})$ of ellipsoidal tube-based MPC: 
\begin{subequations} \label{eq:rob_OCP}
\begin{align}
    \makebox[0pt][c]{$\displaystyle \hspace{5.5 cm}
    \min_{\overline{\vx}_{[0:N]},\,\overline{\vu}_{[0:N-1]},\,\mP_{[1:N]},\,\mK_{[1:N-1]}}
    \ \sum_{k=0}^{N-1} \ell_k(\overline{\vx}_k,\overline{\vu}_k) + V(\overline{\vx}_N)
  $} \\
    \text{s.t.:} \quad & \overline{\vx}_0 = \vx_{\text{init}},\\
     & \overline{\vx}_{k+1} = \vf(\overline{\vx}_k,\overline{\vu}_k,0),\ 0\leq k< N,\\
     & \mP_{k+1} = \psi_k(\mP_k,\vx_k, \mA_k,\mB_k,\mK_k,\Gamma,\sigma),\ 0 \leq k <N,\\
     \begin{split}
     & 0\geq h_k(\overline{\vx}_k,\overline{\vu}_k) \\
     &\quad + \sqrt{H_k(\overline{\vx}_k,\overline{\vu}_k,\mP_k,\mK_k)+\eps} , \ 0 \leq k <N,
     \end{split}\\
     & 0 \geq h_N(\overline{\vx}_N)+ \sqrt{H_N(\overline{\vx}_N,\mP_N)+\eps},
\end{align}
\end{subequations}
where the parameter $\eps$ ensures differentiability at all feasible points by introducing a small tightening of the constraints.
\begin{remark}
    In many cases no constraint violations are observable, even when the nonlinearity bounds are ignored (see~\cite{messererEfficientAlgorithmTubebased2021,gaoStochasticModelPredictive2024,diehlApproximationTechniqueRobust2006}).
\end{remark}
\section{Ellipsoidal Partitioning} \label{sec:Partition}
To introduce a tree structure to the propagation of reachable sets with ellipsoids, we propose to partition the ellipsoidal reachable sets of the tube-based MPC~\eqref{eq:rob_OCP}.
By partitioning the ellipsoidal reachable set $\gE(\mP_1,\overline{\vx}_1)$ via a hyperplane, we introduce recourse by branching two scenarios.
To continue in each scenario with one of the two half-ellipsoids we over-approximate each via the respective Löwner-John ellipsoids, which are the minimum-volume ellipsoids covering the respective half-ellipsoids.
There exists an analytical formula to compute the Löwner-John ellipsoids, which was used in the ellipsoidal method~\cite{blandEllipsoidMethodSurvey1981}.
We describe the halfspace as $\mathcal{H}(\va,b)=\{\vx\in\sR^n | \va^{\intercal} \vx \leq b \}$.
\begin{lemma}[Minimum-Volume Ellipsoid~\cite{blandEllipsoidMethodSurvey1981}] \label{lem:Löwner-John}
    The set $\mathcal{C}$ described by the intersection of an ellipsoid with a halfspace $\gC = \gE(\mQ,\vc)\cap \gH(\va,b)$ is over-approximated by the ellipsoid $\gC\subseteq \gE(\mR,\vd)$ with $\va,\vc,\vd \in \sR^{n}$ and $b\in\sR$:
    \begin{gather} 
        \vd = \vc - \tau\frac{\mQ\va}{\sqrt{\va^{\intercal}\mQ\va}} := \gamma(\mQ,\vc,\va,\alpha(\mQ,\vc,\va,b)), \label{eq:new_center}\\
        \mR = \delta \left( \mQ - \sigma \frac{\mQ\va\va^{\intercal}\mQ^{\intercal}}{\va^{\intercal}\mQ\va} \right) := \rho(\mQ,\vc,\va,\alpha(\mQ,\vc,\va,b)), \label{eq:new_ellipse}
    \end{gather}
    where
    \begin{gather}
    \alpha(\mQ,\vc,\va,b) = \frac{\va^{\intercal}\vc-b}{\sqrt{\va^{\intercal}\mQ\va}},\\
        \tau = \frac{1+n \alpha}{n+1}, \\
        \sigma = \frac{2(1+n\alpha)}{(n+1)(1+\alpha)},\\
        \delta = \frac{n^2}{n^2-1}(1-\alpha^2), \label{eq:delta}
    \end{gather}
    when $-\frac{1}{n}\leq \alpha \leq 1$.
    The relative position of the hyperplane with respect to the center of the ellipsoid is given by $\alpha$.
    For the case that $-1\leq \alpha\leq -\frac{1}{n}$, the minimum-volume ellipsoid enclosing $\gC$ is $\gE(\mQ,\vc)$. If $\alpha\leq -1$, the halfspace $\gH$ intersects with the whole ellipsoid. If $\alpha =1$, the intersection is a singular point and if $\alpha>1$, the intersection is empty.
\end{lemma}

For recursive feasibility, it will be important that the case $\alpha=1$ can be achieved (see proof in Section~\ref{sec:Ell_MS}). However, this means that for the other side of the partition $\alpha=-1$. The formulas~\eqref{eq:new_center} and \eqref{eq:new_ellipse} only hold for $-\frac{1}{n}\leq \alpha \leq 1$. For the case $\alpha=-\frac{1}{n}$, the original ellipsoid is recovered, which is also the minimum-volume ellipsoid for $\alpha<-\frac{1}{n}$
Therefore, we propose to replace $\alpha(\mQ,\vc,\va,b)$ in Lemma~\ref{lem:Löwner-John} by $\overline{\alpha}$ with
\begin{align} \label{eq:overline_alpha}
    \overline{\alpha}(\mQ,\vc,\va,b)=\begin{cases}
        \alpha(\mQ,\vc,\va,b), &\text{ if } \alpha(\mQ,\vc,\va,b)\geq -\frac{1}{n},\\
        -\frac{1}{n} , &\text{ if } \alpha(\mQ,\vc,\va,b) < -\frac{1}{n}.
    \end{cases}
\end{align}
There also exist formulas for multiple partitions~\cite{blandEllipsoidMethodSurvey1981}, but for simplicity we consider only one partition per ellipsoid.
\section{Ellipsoidal multi-stage MPC} \label{sec:Ell_MS}
We now use the multi-stage MPC as in~\cite{luciaMultistageNonlinearModel2013} as a backbone to construct multiple scenarios of ellipsoidal tubes. 
The scenarios are created based on the partitioning process described in Sec.~\ref{sec:Partition}.
So, the first ellipsoid to partition is $\gE_1(\mP_1,\overline{\vx}_1)$, as this partition simulates a future measurement, giving information about the position in the state space at the next time step. The positioning of the hyperplane defining the halfspace $\gH_1(\va_1,\vb_1)$ is a degree of freedom of the optimizer.
Instead of continuing the uncertainty propagation with $\mP_1$ as in~\eqref{eq:rob_OCP} with one tube, the uncertainty is propagated via two tubes continuing from the minimum-volume ellipsoids of both sides of the partition $\gE^1_1(\mP_1^1,\overline{\vx}_1^1)$ and $\gE^2_1(\mP_1^2,\overline{\vx}_1^2)$ with independent inputs $\overline{\vu}_1^1,\overline{\vu}_1^2$, where
\begin{gather*}
    \mP_1^1 = \rho(\mP_1,\overline{\vx}_1,\va_1,\overline{\alpha}(\mP_1,\overline{\vx}_1,\va_1,b_1)),\\
    \overline{\vx}_1^1 =  \gamma(\mP_1,\overline{\vx}_1,\va_1,\overline{\alpha}(\mP_1,\overline{\vx}_1,\va_1,b_1)),\\
    \mP_1^2 = \rho(\mP_1,\overline{\vx}_1,-\va_1,\overline{\alpha}(\mP_1,\overline{\vx}_1,-\va_1,-b_1)),\\
    \overline{\vx}_1^2 =  \gamma(\mP_1,\overline{\vx}_1,-\va_1,\overline{\alpha}(\mP_1,\overline{\vx}_1,-\va_1,-b_1)).
\end{gather*}
This partition can be continued throughout the prediction horizon, meaning that the propagation of $\mP_1^1$, $\psi_1(\mP_1^1,\overline{\vx}_1^1,\mA_1^1,\mB_1^1,\mK_1^1,\Gamma_k^1,\sigma)$ where the linearization occurs around $\overline{\vx}_1^1$ and $\overline{\vu}_1^1$, is divided by another halfspace $\gH(\va_2^1,b_2^1)$. 
Analogously, the propagation of $\mP_1^2$ would get partitioned via $\gH(\va_2^2,b_2^2)$.
This scheme is sketched in Figure~\ref{fig:scheme}.
To avoid the exponential growth of a multi-stage scheme, the branching is stopped after $n_r$ steps.
Thus, after $n_r$, there are $\mu_s=2^{n_r}$ different tubes.
We present the formulation of the proposed ellipsoidal multi-stage MPC $\mathcal{P}^{\text{ElP}}_N(\vx_{\text{init}})$:
\begin{subequations} \label{eq:Ell_part_OCP}
\begin{align}
&\mathclap{\hspace{8 cm}  \min_{\begin{subarray}{c}
\overline{\vx}_{\left[0:N\right]}^{\left[1:\mu_s\right]},
\overline{\vu}_{\left[0:N-1\right]}^{\left[1:\mu_s\right]},
\mP_{\left[1:N\right]}^{\left[1:\mu_s\right]}, \\
\tilde{\vx}_{\left[1:n_r\right]}^{\left[1:\mu_s\right]},
\tilde{\mP}_{\left[1:n_r\right]}^{\left[1:\mu_s\right]}, \\
\mK_{\left[1:N-1\right]}^{\left[1:\mu_s\right]},
\va_{\left[1:n_r\right]}^{\left[1:\mu_s\right]},
b_{\left[1:n_r\right]}^{\left[1:\mu_s\right]} \end{subarray}} \sum_{s=1}^{\mu_s} \sum_{k=0}^{N-1} \omega_k^s \ell_k(\overline{\vx}^s_k, \overline{\vu}_k^s) + \omega_N^s V(\overline{\vx}_N^s)} \notag \\
\label{eq:Ell_part_OCP:cost} \\
&\text{subject to:} \notag \\
&\quad \overline{\vx}_0^s = \vx_{\text{init}}, \quad \forall s \in \sS \label{eq:Ell_part_OCP:x0} \\
& \text{Nominal and uncertainty propagation $\forall k \in \sK_0, \forall s \in \sS$:} \notag \\
&\quad \tilde{\vx}_{k+1}^s = \vf(\overline{\vx}_k^s, \overline{\vu}_k^s, 0) \label{eq:Ell_part_OCP:aux_state} \\
&\quad \tilde{\mP}_{k+1}^s = \psi_k(\mP_k^s, \overline{\vx}_k^s, \mA_k^s, \mB_k^s, \mK_k^s, \Gamma_k^s, \sigma) \label{eq:Ell_part_OCP:aux_P} \\
& \text{Ellipsoidal partitioning $\forall k \in \sK_1, \forall s \in \sS$:} \notag \\
&\quad \overline{\vx}_{k}^s = \gamma(\tilde{\mP}_k^s, \tilde{x}_k^s, \va_k^s, \overline{\alpha}(\tilde{\mP}_k^s, \tilde{x}_k^s, \va_k^s, b_k^s)) \label{eq:Ell_part_OCP:x_partition} \\
&\quad \mP_{k}^s = \rho(\tilde{\mP}_k^s, \tilde{x}_k^s, \va_k^s, \overline{\alpha}(\tilde{\mP}_k^s, \tilde{x}_k^s, \va_k^s, b_k^s)) \label{eq:Ell_part_OCP:P_partition} \\
&\quad -1 \leq \alpha(\tilde{\mP}_k^s, \tilde{x}_k^s, \va_k^s, b_k^s) \leq 1 \label{eq:Ell_part_OCP:part_placement} \\
& \text{Nominal and uncertainty propagation $\forall k \in \sK_r, \forall s \in \sS$:} \notag \\
&\quad \overline{\vx}_{k+1}^s = \vf(\overline{\vx}_k^s, \overline{\vu}_k^s, 0) \label{eq:Ell_part_OCP:x_prop} \\
&\quad \mP_{k+1}^s = \psi_k(\mP_k^s, \overline{\vx}_k^s, \mA_k^s, \mB_k^s, \mK_k^s, \Gamma_k^s, \sigma) \label{eq:Ell_part_OCP:P_prop} \\
& \text{Constraint tightening $\forall k \in \sI_{\left[0:N-1\right]}, \forall s \in \sS$:} \notag \\
&\quad 0 \geq h_k(\overline{\vx}_k^s, \overline{\vu}_k^s) + \sqrt{H_k(\overline{\vx}_k^s, \overline{\vu}_k^s, \mP_k^s, \mK_k^s) + \eps} \label{eq:Ell_part_OCP:stage_constr} \\
&\quad 0 \geq h_N(\overline{\vx}_N^s) + \sqrt{H_N(\overline{\vx}_N^s, \mP_N^s) + \eps} \label{eq:Ell_part_OCP:term_constr} \\
& \text{Non-anticipativity $\forall l \in \sL_u^k, \forall j \in \sJ_u^k$:} \notag \\
&\quad \overline{\vu}_k^l = \overline{\vu}_k^{j+l}, \quad \mK_k^l = \mK_k^{l+j}, \quad \forall k \in \sK_0 \label{eq:Ell_part_OCP:non_anticipativity_u} \\
&\quad \va_k^l = \va_k^{j+l}, \quad b_k^l = b_k^{j+l}, \quad \forall k \in \sI_{\left[1:n_r-1\right]} \label{eq:Ell_part_OCP:non_anticipativity_part} \\
& \text{Halfspace alignment $\forall k \in \sK_1, \forall l \in \sL_a^k, \forall j \in \sJ_a^k$:} \notag \\
&\quad \va_k^l = -\va_k^{j+l}, \quad b_k^l = -\va_k^{j+l} \label{eq:Ell_part_OCP:cont_part}
\end{align}
\end{subequations}
where $\sK_0 = \sI_{\left[0:n_r-1\right]}$, $\sK_1 = \sI_{\left[1:n_r-1\right]},\ \sK_r=\sI_{\left[n_r:N-1\right]}$, $\sS=\sI_{\left[1:\mu_s\right]}$, $\sL_u^k = (\frac{\mu_s}{2^{k}}\sI_{\left[0:k \right]}+1)$, $\sL_a^k = (\frac{\mu_s}{2^{k}}\sI_{\left[0:k \right]}+1)$, $\sJ_u^k= \sI_{\left[1:\frac{\mu_s}{2^{k}}-1 \right]}$, $\sJ_a^k = \sI_{\left[\frac{\mu_s}{2^{k}}+1:\frac{\mu_s}{2^{k-1}}-1 \right]}$.
The weights $\omega_k^s$ in the cost function~\eqref{eq:Ell_part_OCP:cost} are important in a stochastic setting and could represent the probability of the scenario $s$, which is dependent on the placement of the halfspaces.
Instead of a scenario tree, we chose for simplicity of indexing to represent every scenario fully disconnected and constrain branches of the scenario tree in the first steps to be equal via constraints~\eqref{eq:Ell_part_OCP:x0},~\eqref{eq:Ell_part_OCP:non_anticipativity_u} and~\eqref{eq:Ell_part_OCP:non_anticipativity_part}. To decrease the size of the optimization problem, the superfluous variables can be eliminated.
The nominal trajectories are propagated via~\eqref{eq:Ell_part_OCP:aux_state} and~\eqref{eq:Ell_part_OCP:x_prop}. 
For $k<n_r$ an auxiliary variable $\tilde{x}$ is used, as the nominal state is adapted via the partitioning in~\eqref{eq:Ell_part_OCP:x_partition}.
Analogously, the uncertainty is propagated via~\eqref{eq:Ell_part_OCP:aux_P},~\eqref{eq:Ell_part_OCP:P_prop} and partitioned via~\eqref{eq:Ell_part_OCP:P_partition}. The auxiliaries $\tilde{x}_1$ and $\tilde{P}_1$ are used for notational convenience and can be eliminated from the optimization problem.
The stage and terminal constraints in~\eqref{eq:Ell_part_OCP:stage_constr} and~\eqref{eq:Ell_part_OCP:term_constr} need to hold for every scenario.
To ensure causality,~\eqref{eq:Ell_part_OCP:non_anticipativity_u} enforces, that every scenario with the same information needs to have the same input.
Constraint~\eqref{eq:Ell_part_OCP:cont_part} ensures, that both minimum-volume ellipsoids of a partition result from the two sides of the same hyperplane.
Via~\eqref{eq:Ell_part_OCP:part_placement} it is ensured, that the partitioning hyperplane cannot be chosen to lie outside of the respective uncertainty ellipsoid.
The solution of $\mathcal{P}_N^{\text{ElP}}(\vx)$ results in a control law $\kappa^{\text{ElP}}(\vx_{\text{init}})=\overline{\vu}_0^{1*}$.
\begin{figure}
    \centering
    \includegraphics[width=0.49\textwidth]{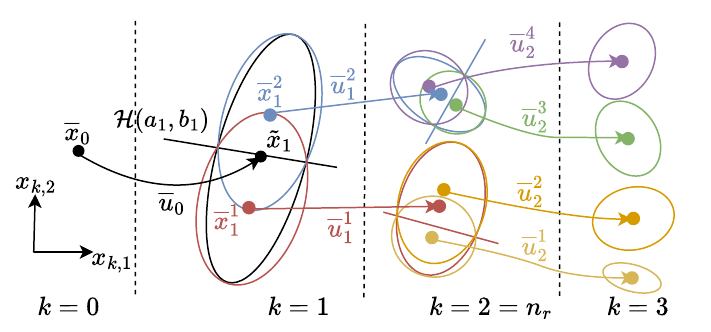}
    \caption{Sketch of the proposed ellipsoidal multi-stage MPC as a state plot over multiple steps $k$ in the prediction horizon. The partitioning is continued until $n_r=2$, afterwards the number of scenarios is kept constant.}
    \label{fig:scheme}
\end{figure}

For recursive feasibility we require following assumptions.
\begin{assumption} \label{ass:non_bound}
    The system~\eqref{eq:system} and the constraints $h_k$ and $h_N$ are affine or the nonlinearity boundaries $\Omega_P,\Omega_{h_k}$ and $\Omega_{h_N}$ over-approximate the linearization error.
\end{assumption}
\begin{assumption} \label{ass:terminal_set}
    The terminal set $\sX_f$ defined via the constraints $\sX_f=\{\vx| h_N(\vx )\leq 0\}$ is robust control invariant under the feedback law $\vu_f = \mK_f\vx + \overline{\vu}_f,$ $\forall w \in \sW$:
    \begin{equation*}
        \forall \vx \in \sX_f, \forall w \in \sW: f(\vx,\vu_f,\vw)\in\sX_f. 
    \end{equation*}
\end{assumption}
We can now prove recursive feasibility.
\begin{theorem}
   Let Assumptions~\ref{ass:non_bound} and \ref{ass:terminal_set} hold.
The MPC problem $\mathcal{P}_N^{\text{ElP}}(\vx)$ is recursively feasible, i.e., for all $\vx_k$, for which $\mathcal{P}_N^{\text{ElP}}(\vx_{k})$ is feasible, $\mathcal{P}_N^{\text{ElP}}(\vx_{k+1})$ is also feasible, where $\vx_{k+1}=\vf(\vx_k,\kappa_N^{\text{ElP}}(\vx_k),\vw_k),\ \vw\in\sW$.
\end{theorem}
\begin{proof}
    To show recursive feasibility, we argue that a new feasible solution $\vu_{\text{new}},\mK_{\text{new}},\va_{\text{new}}, b_{\text{new}}$ can be created from a feasible solution of the previous time step $\vu_{\text{feas}},\mK_{\text{feas}},\va_{\text{feas}}, b_{\text{feas}}$ (the other optimization variables follow from~\eqref{eq:Ell_part_OCP:x0}-\eqref{eq:Ell_part_OCP:P_prop}). 
    For this, it is checked on which side of the hyperplane $\gH(\va_{\text{feas}}, b_{\text{feas}})$ lies the new state. The index of this corresponding scenario is denoted as $j$. 
    We propose the new candidate solution 
    \begin{equation*}
        \vu_{\text{new},k}^{s}=\begin{cases}
            \vu_{\text{feas},k+1}^j ,\ &k=0,...,N-2,\\
            \overline{\vu}_f\, \ &k=N-1,
        \end{cases}, s=1,2
    \end{equation*}
    and analogously for $\mK_{\text{new},k}^s$.
    The partitions can be chosen to align with the time-shifted partitions of the previous solution. 
    The last partition is chosen to lie on the boundary, i.e. $\alpha=1$, to avoid over-approximations of the minimum-volume ellipsoids. This emphasizes the importance of using the reformulation $\overline{\alpha}$, as for the other side of the partition $\alpha=-1$ and~\eqref{eq:new_center} and~\eqref{eq:new_ellipse} would not hold. 
    Due to Assumption~\ref{ass:non_bound} it is guaranteed that the new uncertainty ellipsoids lie inside the uncertainty ellipsoids of the previous solution and do not violate any constraints.
    Assumption~\ref{ass:terminal_set} ensures that an input $u_f$ can be found for the last input in the new solution.
\end{proof}
The assumption of a robust control invariant set is a common assumption in many robust MPC approaches~\cite{mayneRobustModelPredictive2005,subramanianTubeenhancedMultistageMPC2021}. The existence of nonlinearity bounds is an assumption which was often made for ellipsoidal approaches~\cite{houskaRobustOptimizationNonlinear2012,villanuevaUnifiedFrameworkPropagation2015,villanuevaConvexEnclosuresConstrained2019}. Even without the assumption, approximate robust performance and approximate recursive feasibility is achieved.
\section{Collision avoidance example} \label{sec:Case_Study}
In the case study, a robot shall be controlled to traverse a corridor while avoiding a human whose uncertain movement is directed in the opposite direction~\cite{gaoStochasticModelPredictive2024}.
The dynamic models for both the robot and the human are concatenated in the states
\begin{equation}
    \dot{\vx}=\begin{bmatrix}
        \dot{x}_{r,x}\\
        \dot{x}_{r,y}\\
        \dot{x}_{r,\theta}\\
        \dot{x}_{r,v}\\
        \dot{x}_{r,\omega}\\
        \dot{x}_{h,x}\\
        \dot{x}_{h,y}
    \end{bmatrix} = \begin{bmatrix}
        x_{r,v} \cos(x_{r,\theta})\\
        x_{r,v} \sin(x_{r,\theta})\\
        x_{r,\omega} \\
        u_{r,a}\\
        u_{r,\alpha}\\
        v_{h,x}+w_{h,x}\\
        v_{h,y}+w_{h,y}\\
    \end{bmatrix},
\end{equation}
where these differential equations are discretized numerically via an explicit Runke-Kutta integrator of fourth order with the time discretization $\Delta t = 0.5 s$.
The velocities $v_{h,x}$ and $v_{h,y}$ are known, possibly time-varying piece-wise constant parameters. In this study, we set $v_{h,x}=-0.6\ m/s$ and $v_{h,y}=0\ m/s$.
The unknown disturbance $\vw=\left[w_{h,x}\ w_{h,y} \right]^{\intercal}$ is assumed to be piecewise constant with respect to the discretization interval $\Delta t$ and belonging to an ellipsoidal set $\vw\in\gE(0.4^2 I,0)$.
The distance constraint
\begin{equation} \label{eq:safety_distance}
    \left|\left| \left[x_{r,x}\ x_{r,y} \right]^{\intercal} - \left[x_{h,x}\ x_{h,y} \right]^{\intercal}\right|\right|_2 - \Delta_{\text{safe}}\geq 0,
\end{equation}
with $\Delta_{\text{safe}}=0.3$ accounts for the combined radii of the robot and human.
All remaining state constraints are box constraints
\begin{equation}
    \begin{bmatrix}
        -1.3\\
        0\\
        -1
    \end{bmatrix} \leq
    \begin{bmatrix}
         x_{r,y}\\
        x_{r,v}\\
        x_{r,\omega}\\
    \end{bmatrix} \leq \begin{bmatrix}
        1.3\\
        2\\
        1
    \end{bmatrix}.
\end{equation}
The terminal constraints additionally include the constraint that $x_{r,v}\leq 0.01$ to ensure a safe resting position at the end of the prediction horizon~\cite{gaoStochasticModelPredictive2024}.
The objective of the robot is to travel at near-constant speed through the middle of the corridor, so the objective is
\begin{equation}
    \ell_k(\vx_k,\vu_k)=(\vx_k-\vx_k^{\text{ref}})^{\intercal}\mQ (\vx_k-\vx_k^{\text{ref}}) + \vu_k^{\intercal}\mR \vu_k,
\end{equation}
with $\mQ= \text{diag}(\left[50,50,0,2,0,0,0\right])$ and $\mR=2\mI$ and $x_k^{\text{ref}}=\left[x_{\text{init},r,x}+v_x^{\text{ref}}k\Delta t, 0,0 , v_x^{\text{ref}} , 0 , 0 , 0\right]^{\intercal}$ with $v_x^{\text{ref}} = 1.5\ m/s$.
For every experiment, the initial state is $\vx_{\text{init}}=\left[-2 , 0, 0 , 1.6 , 0 , -3.5 , 0 \right]^{\intercal}$.
All experiments were performed on a desktop PC with \emph{Intel i9-13900K} CPU. The optimization problem was solved via \emph{Ipopt}~\cite{wachterImplementationInteriorpointFilter2006} in the CasADi framework~\cite{anderssonCasADiSoftwareFramework2019} with \emph{MA-27} (HSL-Subroutines~\url{http://www.hsl.rl.ac.uk}). The code is openly available\footnote{\url{https://github.com/MoritzHein/Ellipsoid_Partition}\label{URL}}.
\subsection{Comparison of predictions}
We apply the proposed ellipsoidal multi-stage approach~\eqref{eq:Ell_part_OCP} with one partition ($n_r=1$, so $\mu_s=2$). To demonstrate the effect of the partitioning, $\mK_k^s$ were set to $\vzero$. The weights $\omega_k^s$ in the cost function are set to 1. To implement~\eqref{eq:overline_alpha} in the gradient-based framework, it was conservatively smoothed via an adapted softplus function
\begin{equation} \label{eq:softplus}
    \overline{\alpha}_{\text{smooth}}(\alpha)=\frac{\log \left( 1+\nu \exp(\eta \beta (\alpha+\frac{1}{n_x})\right)}{\beta}-\frac{1}{n_x},
\end{equation}
where $\beta$ determines the smoothing of the function, while $\eta$ and $\nu$ are fitted to achieve a minimal offset at $\alpha=-1/n_x$.
\begin{figure}
    \centering
    \includegraphics[width=\linewidth]{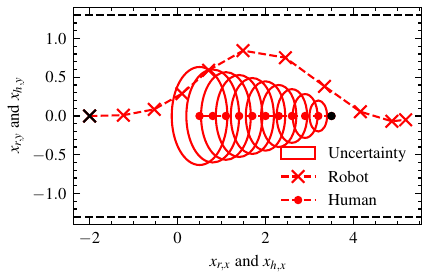}
    \caption{Open-loop prediction of~\eqref{eq:rob_OCP} at the initial state as a position plot. The robot drives from left to right, while the human walks from right to left.}
    \label{fig:Pred_No_Scen}
\end{figure}

The optimal solution of the standard ellipsoidal tube-based method~\eqref{eq:rob_OCP} with a prediction horizon of $N=10$ and $\mK=0$ is shown in Figure~\ref{fig:Pred_No_Scen} in a position plot.
Figure~\ref{fig:OL_pred} presents the optimal solution of the proposed ellipsoidal multi-stage MPC~\eqref{eq:Ell_part_OCP} for the initial state. 
The partition obtained as a result of the proposed method is horizontal, enabling the robot to decide after the first time step whether to turn left or right depending on the realization of the uncertainty.
\begin{figure}
    \centering
    \includegraphics[width=\linewidth]{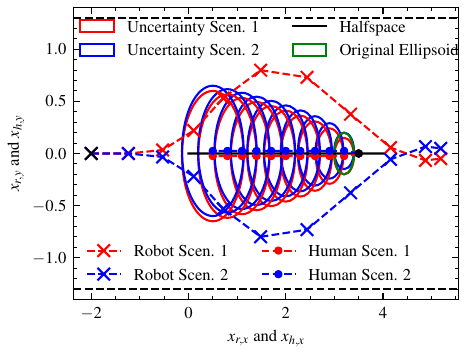}
    \caption{Open-loop prediction of~\eqref{eq:Ell_part_OCP} at the initial state as a position plot. The green ellipsoid is partitioned into the red and blue ellipsoids, which are further propagated. The robot drives from left to right. The human walks from right to left.}
    \label{fig:OL_pred}
\end{figure}
The advantage of the proposed method compared to a standard ellipsoidal propagation following problem~\eqref{eq:rob_OCP} is that~\eqref{eq:rob_OCP} can only consider one trajectory, so it initially decides to either go above or below the human (in Figure~\ref{fig:Pred_No_Scen} it goes up). This still leads to a safe behavior, but is more conservative when the initial decision is wrong so the robot steers farther from the reference. 
\subsection{Comparison in closed-loop}
To further illustrate the advantages of the proposed ellipsoidal multi-stage MPC, we compare the closed-loop behavior of the proposed approach~\eqref{eq:Ell_part_OCP} with the standard ellipsoidal approach~\eqref{eq:rob_OCP} with $\mK=0$, as well as with $\mK_k^s$ as degrees of freedom for the optimizer in the partial feedback setting of~\cite{gaoStochasticModelPredictive2024} (only feedback on $x_{h,x}, x_{h,y}$ and $x_{r,v}$).
\begin{figure*}
\centering
\begin{subfigure}{0.25 \textwidth}
    \includegraphics[width=1\textwidth]{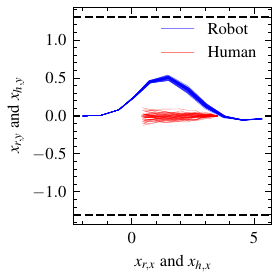}
    \caption{Ellipsoidal tube-based MPC\\
    $\mK=0$}
    \label{fig:CL_Ell_noK}
\end{subfigure}%
\begin{subfigure}{0.25 \textwidth}
    \includegraphics[width=1\textwidth]{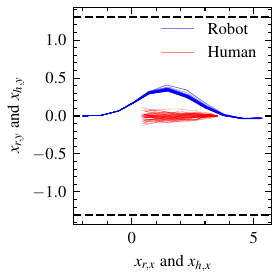}
    \caption{Ellipsoidal tube-based MPC\\
    $\mK$ optimized}
    \label{fig:CL_Ell_K}
\end{subfigure}%
\begin{subfigure}{0.25 \textwidth}
    \includegraphics[width=1\textwidth]{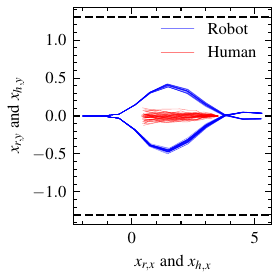}
    \caption{Ellipsoidal multi-stage MPC\\
    $\mK=0$}
    \label{fig:CL_Part_noK}
\end{subfigure}%
\begin{subfigure}{0.25 \textwidth}
    \includegraphics[width=1\textwidth]{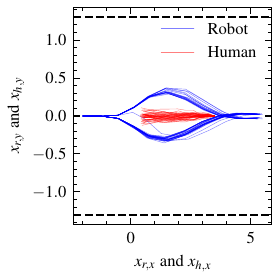}
    \caption{Ellipsoidal multi-stage MPC\\
    $\mK$ optimized}
    \label{fig:CL_Part_K}
\end{subfigure}%
    \caption{50 closed-loop trajectories of the robot and the human in a position plot of different robust MPC approaches.}
    \label{fig:CL}
\end{figure*}
Figure~\ref{fig:CL} shows the closed-loop behavior of the four robust MPC algorithms tested on 50 random parameter realizations. Table~\ref{tab:CLC_comp} compares the corresponding mean closed-loop costs as well as the computation times. As suggested in the previous subsection, the standard ellipsoidal tube-based MPC~\eqref{eq:rob_OCP} always chooses the same direction for the robot to avoid the human. 
The partitioning strategy of the proposed ellipsoidal multi-stage leads to the robot choosing both the path above or below the human. This improves the mean closed-loop cost by around $19\%$. However, the computation takes more time due to the increase in optimization variables.
Using the partial feedback strategy in~\cite{gaoStochasticModelPredictive2024} improves the performance for both approaches, as the feedback policy allows the robot to go closer to the human. The solver is initialized with the non-feedback solution and $\mK_k^s$ were slightly regularized. The additional non-convexity of optimizing over the feedback policies lead to a large increase in computation time for both approaches.
No approach yields a constraint violation.
\begin{table}
    \centering
    \begin{tabular}{llllll} \toprule
    \multirow{2}{*}{\begin{tabular}[x]{@{}c@{}}Performance\\metric\end{tabular}} & \multicolumn{2}{c}{Ellipsoidal tube} && \multicolumn{2}{c}{Ellipsoidal multi-stage}\\
    \cmidrule{2-3} \cmidrule{5-6}
          & $\mK=0$ & $\mK$ opt. && $\mK=0$ & $\mK$ opt. \\
         \midrule
       Closed-loop cost  & 43.95  & 24.21  &&    35.63 & 20.79  \\
        CPU Time [s] & 0.137  & 1.490 && 0.990 & 36.498\\
        \bottomrule
    \end{tabular}
    \caption{Mean closed-loop cost and computation times in seconds per solution over 50 random parameter realizations.}
    \label{tab:CLC_comp}
\end{table}
If the optimizer did not find a feasible solution, the respective input from the last previously feasible solution is taken.
\section{Conclusion} \label{sec:Conclusion}
We present an ellipsoidal multi-stage MPC for robust optimal control.
The presented approach combines the advantages of a multi-stage scheme with the rigorous handling of reachable sets of ellipsoidal tube-based schemes, so that multiple independent reachability tubes can be considered. This leads to reduced conservatism in settings like object avoidance, which is demonstrated in a human-robot collision avoidance case study.
Further steps involve improving real-time feasibility using zero-order optimization schemes.
\bibliographystyle{IEEEtran}
\bibliography{IEEEabrv,PhD_abr.bib}

\end{document}